\documentclass[UKenglish,11pt]{article}
\input{preamble.tex}
\usepackage[margin=1in]{geometry}

\newcommand{\kbp}{\texttt{$k$-BAL\-ANCED PAR\-TI\-TION\-ING}\xspace}
\newcommand{\bisect}{\texttt{BI\-SEC\-TION}\xspace}

\begin{document}

\title{Fast Balanced Partitioning is Hard\\
Even on Grids and Trees\footnote{A preliminary version of this article appeared
at the \emph{37th International Symposium on Mathematical Foundations of
Computer Science}~\cite{AF12}}}

\date{}

\author{Andreas Emil Feldmann}
\affil{Combinatorics and Optimization Department,

University of Waterloo, Canada

\href{mailto:feldmann.a.e@gmail.com}
{\texttt{feldmann.a.e@gmail.com}}
}

\maketitle

\begin{abstract}
Two kinds of approximation algorithms exist for the \kbp problem: those that are
fast but compute unsatisfactory approximation ratios, and those that guarantee
high quality ratios but are slow. In this article we prove that this tradeoff
between \emph{running time} and \emph{solution quality} is unavoidable. For the
problem a minimum number of edges in a graph need to be found that, when cut,
partition the vertices into $k$ equal-sized sets. We develop a general reduction
which identifies some sufficient conditions on the considered graph class in
order to prove the hardness of the problem. We focus on two combinatorially
simple but very different classes, namely \emph{trees} and \emph{solid grid
graphs}. The latter are finite connected subgraphs of the infinite
two-dimensional grid without holes.
We apply the reduction to show that for solid grid graphs it is NP-hard to
approximate the optimum number of cut edges within any satisfactory ratio. We
also consider solutions in which the sets may deviate from being equal-sized.
Our reduction is applied to grids and trees to prove that no \emph{fully
polynomial time} algorithm exists that computes solutions in which the sets are
arbitrarily close to equal-sized. This is true even if the number of edges cut
is allowed to increase when the limit on the set sizes decreases. These
are the first bicriteria inapproximability results for the \kbp problem. 
\end{abstract}


\section{Model and Setting}

We consider the \kbp problem in which the $n$ vertices of a graph need to be
partitioned into $k$ sets of size at most $\lceil n/k\rceil$ each. At the same
time the \emph{cut size}, which is the number of edges connecting vertices from
different sets, needs to be minimised. This problem has many applications
including 
VLSI circuit design~\cite{BhattL84},
image processing~\cite{LeahyW93},
computer vision~\cite{KwatraSETB03},
route planning~\cite{DellingGPW11},
and divide-and-conquer algorithms~\cite{LiptonT80}.
In our case the motivation (cf.~\cite[Section~4]{ArbenzLMMS07}) stems from
parallel computations for finite element models (FEMs). In these a continuous
domain of a physical model is discretised into a mesh of sub-domains (the
elements). The mesh induces a graph in which the vertices are the elements and
each edge connects neighbouring sub-domains. A vertex then
corresponds to a computational task in the physical simulation, during which
tasks that are adjacent in the graph need to exchange data. Since the model is
usually very large, the computation is done in parallel.
Hence the tasks need to be scheduled on to $k$ machines (which corresponds to a
partition of the vertices) so that the loads of the machines (the sizes of the
sets in the partition) are balanced.
At the same time the inter-processor communication (the cut size) needs to be
minimised since this constitutes a bottleneck in parallel computing. Although
we will give some conclusion for more general cases, in this article we mainly
focus on 2D FEMs. For these the corresponding graph is a planar graph, typically
given by a triangulation or a quadrilateral tiling of the
plane~\cite{ElmanSW05}. We concentrate on the latter and consider so called
\emph{solid grid graphs} which model tessellations into squares. A \emph{grid
graph} is a finite subgraph of the infinite two-dimensional grid. An interior
face of a grid graph is called a \emph{hole} if more than four edges surround
it. If a grid graph is connected and does not have any holes, it is called
\emph{solid}.

In general it is NP-hard to approximate the cut size of \kbp within any finite
factor~\cite{RackeA06}. However the corresponding reduction relies on the fact
that a general graph may not be connected and thus the optimal cut size can be
zero. Since a 2D FEM always induces a connected planar graph, this strong
hardness result may not apply. Yet even for trees~\cite{AFFoschini12} it is
NP-hard to approximate the cut size within~$n^c$, for any constant $c<1$. The
latter result however relies on the fact that the maximum degree of a tree can
be arbitrarily large. Typically though, a 2D FEM induces a graph of constant
degrees, as for instance in grid graphs. In fact, even though approximating the
cut size in constant degree trees is APX-hard~\cite{AFFoschini12}, there exists
an $\mc{O}(\log(n/k))$ approximation algorithm~\cite{MacGregor1978}
for these. This again raises the question of whether efficient approximation
algorithms can be found for graphs induced by 2D FEMs. In this article we give a
negative answer to this question. We prove that it is NP-hard to approximate the
cut size within $n^c$ for any constant $c<1/2$ for solid grid graphs. We also
show that this is asymptotically tight by providing a corresponding
approximation algorithm.

Hence when each set size is required to be at most $\lceil n/k\rceil$ (the
\emph{perfectly balanced} case), the achievable approximation factors are not
satisfactory. To circumvent this issue, both in theory and practice it has
proven beneficial to consider \emph{bicriteria approximations}. Here
additionally the sets may deviate from being perfectly balanced. The computed
cut size is compared with the optimal perfectly balanced solution. Throughout
this article we denote the approximation ratio on the cut size by~$\alpha$.

For planar graphs the famous Klein-Plotkin-Rao Theorem~\cite{KleinPR93} can be
combined with spreading metric techniques~\cite{EvenNRS99} in order to compute a
solution for which $\alpha\in\mc{O}(1)$ and each set has size at most $2\lceil
n/k\rceil$. This needs $\mc{\widetilde O}(n^3)$ time or $\mc{\widetilde O}(n^2)$
expected time. For the same guarantee on the set sizes, a faster algorithm
exists for solid grid graphs~\cite{AF12thesis}. It runs in $\mc{\widetilde
O}(n^{1.5})$ time but approximates the cut size within $\alpha\in\mc{O}(\log
k)$. However it is not hard to see how set sizes that deviate by a factor of $2$
from being perfectly balanced may be undesirable for practical applications. For
instance in parallel computing this means a significant slowdown. This is why
graph partitioning heuristics such as Metis~\cite{KarypisK95} or
Scotch~\cite{ChevalierP08} allow to compute \emph{near-balanced} partitions.
Here each set has size at most $(1+\eps)\lceil n/k\rceil$, for arbitrary
$\eps>0$. The heuristics used in practice however do not give any guarantees on
the cut size. For general graphs the best algorithm~\cite{AFFoschini12} known
that gives such a guarantee, will compute a near-balanced solution for which
$\alpha\in\mc{O}(\log n)$. However the running time of this algorithm increases
exponentially with decreasing $\eps$s. Therefore this algorithm is too slow for
practical purposes. Do algorithms exist that are both \emph{fast} and compute
\emph{near-balanced} solutions, and for which rigorous guarantees can be given
on the computed cut size? Note that the factor $\alpha$ of the above algorithm
does not depend on $\eps$. It therefore suggests itself to devise an algorithm
that will compensate the cost of being able to compute near-balanced solutions
not in the running time but in the cut size (as long as it does not increase too
much). In this article however, we show that no such algorithm exists that is
reasonable for practical applications. More precisely, we consider \emph{fully
polynomial time} algorithms for which the running time is polynomial in
$n/\eps$, for a value $\eps>0$ that is part of the input. We show that, unless
P=NP, for solid grid graphs there is no such algorithm for which the computed
solution is near-balanced and $\alpha=n^c/\eps^d$, for any constants $c$ and $d$
where $c<1/2$.

Our main contribution is a general reduction with which hardness results such as
the two described above can be generated. For it we identify some sufficient
conditions on the considered graphs which make the problem hard. Intuitively
these conditions entail that cutting vertices from a graph must be expensive in
terms of the number of edges used. We also apply the proposed reduction to
general graphs and trees, in order to complement the known results. For general
(disconnected) graphs we can show that, unless P=NP, there is no finite value
for $\alpha$ allowing a fully polynomial time algorithm that computes
near-balanced partitions. For trees we can prove that this is true for any
$\alpha=n^c/\eps^d$, for arbitrary constants $c$ and $d$ where~$c<1$. These
results demonstrate that the identified sufficient conditions capture a
fundamental trait of the \kbp problem. In particular since we prove the hardness
for two combinatorially simple graph classes which however are very dissimilar
(as for instance documented by the high \emph{tree-width} of solid
grids~\cite{DiestelJGT99}). For solid grid graphs we harness their isoperimetric
properties in order to satisfy the conditions, while for trees we use their
ability to have high vertex degrees instead. These are the first bicriteria
inapproximability results for the problem. We also show that all of them are
asymptotically tight by giving corresponding approximation algorithms.

\paragraph*{Related Work.}
Apart from the results mentioned above, Simon and Teng~\cite{SimonT97} gave a
framework with which bicriteria approximations to \kbp can be computed. It is a
recursive procedure that repeatedly uses a given algorithm for \emph{sparsest
cuts}. If a sparsest cut can be approximated within a factor of $\beta$ then
their algorithm obtains ratios~$\eps=1$ and $\alpha\in\mc{O}(\beta\log k)$.
The best factor $\beta$ for general graphs~\cite{AroraRV04} is
$\mc{O}(\sqrt{\log n})$. For planar graphs Park and Phillips~\cite{ParkP93}
show how to obtain $\beta\in\mc{O}(t)$ in $\mc{\widetilde O}(n^{1.5+1/t})$ time,
for arbitrary $t$. On solid grid graphs constant approximations to sparsest cuts
can be computed in linear time~\cite{AFDasWidmayer11}. For general graphs the
best ratio~$\alpha$ is achieved by Krauthgamer {\it et
al.}~\cite{KrauthgamerNS09}. For $\eps=1$ they give an algorithm for which
$\alpha\in\mc{O}(\sqrt{\log n \log k})$.

Near-balanced partitions were considered by Andreev and R\"acke~\cite{RackeA06}
who showed that a ratio of $\alpha\in\mc{O}(\log^{1.5}(n)/\eps^2)$ is possible.
This was later improved~\cite{AFFoschini12} to $\alpha\in\mc{O}(\log n)$, making
$\alpha$ independent of $\eps$. In the latter paper also a PTAS is given for
trees. For perfectly balanced solutions, there is an approximation algorithm
achieving $\alpha\in\mc{O}(\Delta\log_\Delta(n/k))$ for
trees~\cite{MacGregor1978}, where $\Delta$ is the maximum degree.

The special case when $k=2$ (the \bisect problem) has been thoroughly studied.
The problem is NP-hard in general~\cite{GareyJS76} and can be approximated
within $\mc{O}(\log n)$~\cite{Racke08}. Assuming the Unique Games Conjecture, no
constant approximations are possible in polynomial time~\cite{KhotV05}. Also,
unless NP$\subseteq\cap_{\epsilon > 0}$ BPTIME($2^{n^\epsilon}$), no PTAS
exists for this problem~\cite{Khot06}.
Leighton and Rao~\cite{LeightonR99} show how near-balanced solutions for which
$\alpha\in\mc{O}(\beta/\eps^3)$ can be computed, where $\beta$ is as above. In
contrast to the case of arbitrary~$k$, the \bisect problem can be computed
optimally in $\mc{O}(n^4)$ time for solid grid graphs~\cite{AFWidmayer11}, and
in $\mc{O}(n^2)$ time for trees~\cite{MacGregor1978}. For planar graphs the
complexity of \bisect is unknown.

\newcommand{\tp}{\texttt{3-PARTITION}\xspace}

\section{A General Reduction}

To derive the hardness results we give a reduction from the \tp problem defined
below. It is known that \tp is strongly NP-hard~\cite{GareyJ79} which means that
it remains so even if all integers are polynomially bounded in the size of the
input.

\begin{dfn}[\tp]\label{dfn:tp}
Given $3k$ integers $a_1,\ldots,a_{3k}$ and a threshold~$s$ such that
$s/4<a_i<s/2$ for each $i\in\{1,\ldots 3k\}$, and $\sum_{i=1}^{3k} a_i=ks$, find
a partition of the integers into $k$ triples such that each triple sums up to
exactly $s$.
\end{dfn}

We will set up a general reduction from \tp to different graph
classes. This will be achieved by identifying some structural properties that
a graph constructed from a \tp instance has to fulfil, in order to show the
hardness of the \kbp problem. We will state a lemma which asserts that if the
constructed graph has these properties then an algorithm computing near-balanced
partitions and approximating the cut size within some $\alpha$, is able to
decide the \tp problem. We will see that carefully choosing the involved
parameters for each of the given graph classes yields the desired reductions.
While describing the structural properties we will exemplify them for general
(disconnected) graphs which constitute an easily understandable case. For these
graphs it is NP-hard to approximate the cut size within any finite
factor~\cite{RackeA06}. We will show that, unless P=NP, no fully polynomial time
algorithm exists for any $\alpha$ when near-balanced solutions are desired.

For any \tp instance we construct $3k$ graphs, which we will call
\emph{gadgets}, with a number of vertices proportional to the integers $a_1$
to~$a_{3k}$. In particular, for general graphs each gadget $G_i$, where
$i\in\{1,\ldots 3k\}$, is a connected graph on $2a_i$ vertices. This assures
that the gadgets can be constructed in polynomial time since \tp is strongly
NP-hard. In general we will assume that we can construct $3k$ gadgets for the
given graph class such that each gadget $G_i$ has $pa_i$ vertices for some
parameter $p$ specific to the graph class. These gadgets will then be connected
using some number $m$ of edges. The parameters $p$ and $m$ may depend on the
values of the given \tp instance. For the case of general graphs we chose $p=2$
and we let $m=0$, i.e.\ the gadgets are disconnected. In order to show that the
given gadgets can be used in a reduction, we require that an upper bound can be
given on the number of vertices that can be cut out using a limited number of
edges. More precisely, given any colouring of the vertices of all gadgets
into~$k$ colours, by a \emph{minority vertex} in a gadget $G_i$ we mean a vertex
that has the same colour as less than half of $G_i$'s vertices. Any partition of
the vertices of all gadgets into $k$ sets induces a colouring of the vertices
into $k$ colours. For approximation ratios $\alpha$ and $\eps$, the property we
need is that cutting the graph containing $n$ vertices into $k$ sets using at
most $\alpha m$ edges, produces less than~$p-\eps n$ minority vertices in total.
Clearly $\eps$ needs to be sufficiently small so that the graph exists. When
considering fully polynomial time algorithms,~$\eps$ should however also not be
too small since otherwise the running time may not be polynomial. For general
graphs we achieve this by choosing $\eps=(2ks)^{-1}$. This means that $p-\eps
n=1$ since $n=\sum_{i=1}^{3k}pa_i=2ks$. Simultaneously the running time of a
corresponding algorithm is polynomial in the size of the \tp instance since \tp
is strongly NP-hard. Additionally the desired condition is met for this graph
class since no gadget can be cut using $\alpha m=0$ edges. The following
definition formalises the needed properties.

\begin{dfn}\label{dfn:red-set}
For each instance $I$ of \tp with integers $a_1$ to $a_{3k}$ and threshold $s$,
a \emph{reduction set for \kbp} contains a graph determined by some given
parameters $m\geq 0$, $p\geq 1$, $\eps\geq 0$, and $\alpha\geq 1$ which may
depend on $I$. Such a graph constitutes $3k$ (disjoint) gadgets connected
through $m$ edges. Each gadget $G_i$, where $i\in\{1,\ldots,3k\}$, has $pa_i$
vertices. Additionally, if a partition of the $n$ vertices of the graph into $k$
sets has a cut size of at most $\alpha m$, then in total there are less than
$p-\eps n$ minority vertices in the induced colouring.
\end{dfn}

Obviously the involved parameters have to be set to appropriate values in order
for the reduction set to exist. For instance $p$ must be an integer and~$\eps$
must be sufficiently small compared to $p$ and $n$. Since however the
values will vary with the considered graph class we fix them only later. In the
following lemma we will assume that the reduction set exists and therefore all
parameters were chosen appropriately. It assures that given a reduction set, a
bicriteria approximation algorithm for \kbp can decide the \tp problem. For
general graphs we have seen above that a reduction set exists for any finite
$\alpha$ and $\eps=(2ks)^{-1}$. This means that a fully polynomial time
algorithm for \kbp computing near-balanced partitions and approximating the cut
size within $\alpha$, can decide the \tp problem in polynomial time. Such an
algorithm can however not exist, unless P=NP.

\begin{lem}\label{lem:reduction}
Let an algorithm $\mc{A}$ be given that for any graph in a reduction set for
\kbp computes a partition of the $n$ vertices into $k$ sets of size at most
$(1+\eps)\lceil n/k \rceil$ each. If the cut size of the computed solution
deviates by at most $\alpha$ from the optimal cut size of a perfectly balanced
solution, then the algorithm can decide the \tp problem.
\end{lem}
\begin{proof}
Let $k$ be the value given by a \tp instance $I$, and let $G$ be the graph
corresponding to $I$ in the reduction set. Assume that $I$ has a
solution. Then obviously cutting the $m$ edges connecting the gadgets of~$G$
gives a perfectly balanced solution to $I$. Hence in this case the
optimal solution has cut size at most $m$. Accordingly algorithm
$\mc{A}$ will cut at most~$\alpha m$ edges since it approximates the cut
size by a factor of $\alpha$. We will show that in the other case when $I$
does not have a solution, the algorithm cuts more than $\alpha m$
edges. Hence $\mc{A}$ can decide the \tp problem and the lemma follows.

For the sake of deriving a contradiction, assume that algorithm $\mc{A}$ cuts at
most $\alpha m$ edges in case the \tp instance $I$ does not permit a solution.
Since the corresponding graph $G$ is from a reduction set for \kbp, by
Definition~\ref{dfn:red-set} this means that from its $n$ vertices, in total
less than $p-\eps n$ are minority vertices in the colouring induced by the
computed solution of $\mc{A}$. Each gadget $G_i$, where $i\in\{1,\ldots 3k\}$,
of $G$ has a \emph{majority colour}, i.e.\ a colour that more than half the
vertices in $G_i$ share. This is because the size of $G_i$ is $pa_i$ and we can
safely assume that $a_i\geq 2$ (otherwise the instance $I$ is trivial due to
$s/4<a_i<s/2$). The majority colours of the gadgets induce a partition $\mc{P}$
of the integers $a_i$ of $I$ into $k$ sets. That is, we introduce a set in
$\mc{P}$ for each colour and put an integer $a_i$ in a set if the majority
colour of $G_i$ equals the colour of the set.

Since we assume that $I$ does not admit a solution, if every set in
$\mc{P}$ contains exactly three integers there must be some set for which the
contained integers do not sum up to exactly the threshold $s$. On the
other hand the bounds on the integers, assuring that $s/4 < a_i <
s/2$ for each $i\in\{1,\ldots, 3k\}$, mean that in case not every set in
$\mc{P}$ contains exactly three elements, there must also exist a set for
which the contained numbers do not sum up to $s$. By the pigeonhole
principle and the fact that the sum over all $a_i$ equals $ks$, there
must thus be some set $T$ among the $k$ in $\mc{P}$ for which the sum of
the integers is strictly less than $s$. Since the involved numbers are
integers we can conclude that the sum of the integers in $T$ is in fact at
most $s-1$. Therefore the number of vertices in the gadgets corresponding
to the integers in $T$ is at most $p(s-1)$. Let w.l.o.g.\ the
colour of $T$ be $1$. Apart from the vertices in these gadgets having
majority colour~$1$, all vertices in $G$ that also have colour~$1$ must be
minority vertices. Hence there must be less than $p(s-1)+p-\eps n$
many vertices with colour~$1$. Since $\sum_{i=1}^{3k} a_i=ks$ and
thus~$p s=n/k$, these are less than $n/k-\eps n$.

At the same time the algorithm computes a solution inducing a colouring in which
each colour has at most $(1+\eps)n/k$ vertices, since $n=p ks$ is
divisible by $k$. This means we can give a lower bound of $n-(k-1)(1+\eps)n/k$
on the number of vertices of a colour by assuming that all other colours have
the maximum number of vertices. Since this lower bound equals $(1+\eps)n/k-\eps
n$, for any $\eps\geq 0$ we get a contradiction on the upper bound derived above
for colour $1$. Thus the assumption that the algorithm cuts less than $\alpha
m$ edges if~$I$ does not have a solution is wrong.
\end{proof}

\ignore{
For any \tp instance we construct a set of $3k$ graphs, which we will call
\emph{gadgets}, with a number of vertices proportional to the integers $a_1$ to
$a_{3k}$. In particular, for general graphs each gadget $G_i$, where
$i\in\{1,\ldots 3k\}$, is a connected graph on $(\eps n+1) a_i$ vertices, where
$n$ is the total number of vertices in all gadgets. Clearly we need to assume
that $\eps$ is sufficiently small in order for such a set of gadgets to exist.
This can be achieved by letting $\eps=(2ks)^{-1}$ since the total number of
vertices is $n=\sum_{i=1}^{3k} (\eps n+1)a_i=(\eps n+1) ks$. Solving this for
$n$ and setting the value of $\eps$ accordingly gives $n=2ks$. At the same time
this assures that the gadgets can be constructed in polynomial time since \tp is
strongly NP-hard. For the same reason a fully polynomial time algorithm for \kbp
will have polynomial runtime in the size of the \tp instance for this value of
$\eps$. We will show that such an algorithm can decide the \tp problem for any
approximation ratio $\alpha$ on the cut size.

In general we will assume that we can construct $3k$ gadgets for the given graph
class such that each gadget has $pa_i$ vertices for some parameter $p$ specific
for the graph class. These gadgets will then be connected using some number $m$
of edges. The parameters $p$ and $m$ may depend on the values of the given \tp
instance. For the case of general graphs $m=0$, i.e.\ the gadgets are
disconnected. In order to show that the given gadgets can be used in a
reduction, we will assume that they have the property that an upper bound can be
given on the number of vertices that can be cut out using a limited number of
edges. More precisely given any colouring of the vertices of all gadgets into
$k$ colours, by a \emph{minority vertex} in a gadget $G_i$ we mean a vertex that
has the same colour as less than half of $G_i$'s vertices. Any partition of the
vertices of all gadgets into $k$ sets induces a colouring of the vertices into
$k$ colours. The property we need is that using at most $\alpha m$ edges, for
some approximation ratio $\alpha$, to cut the graph into $k$ sets produces less
than $p-\eps n$ minority vertices in total. For general graphs this is easily
established since $\alpha m=0$ for any finite $\alpha$ and hence no gadget can
be cut in this case. The following definition formalises the properties needed
for our reductions.

\begin{dfn}\label{dfn:red-set}
A \emph{reduction set for \kbp} contains a graph for every instance
of \tp, for some given parameters $m$, $p$, $\eps$, and $\alpha$ which may
depend on the values of the instances. Given such an instance the corresponding
graph constitutes $3k$ gadgets connected through $m$ edges such that a gadget
$G_i$, where $i\in\{1,\ldots,3k\}$, has $pa_i$ vertices. Additionally if the
cut size of a partition of the $n$ vertices in such a graph into $k$ sets is at
most $\alpha m$, then in total there are less than $p-\eps n$ minority vertices
in the induced colouring.
\end{dfn}

Obviously the involved parameters have to be set to appropriate values in order
for the reduction set to exist. However the values will vary with the considered
graph class and we therefore fix them only later. The following lemma assures
that given a reduction set an algorithm for \kbp computing near-balanced
partitions and approximating the cut size within $\alpha$ can decide the \tp
problem. For general graphs we have seen above that a reduction set exists for
any finite $\alpha$. This means that a fully polynomial time algorithm for \kbp
approximating the cut size within $\alpha$ can decide the \tp problem within a
polynomial time in the size of a \tp instance. Such an algorithm can however not
exist unless P=NP.

\begin{lem}\label{lem:reduction}
Let for $\eps\geq 0$ an algorithm be given that computes a partition of the
vertices for any graph in a reduction set for \kbp into $k$ sets of size
at most $(1+\eps)\lceil n/k \rceil$ each. If the cut size of the computed
solution deviates by at most $\alpha$ from the optimal cut size of a perfectly
balanced solution, then the algorithm can decide the \tp problem.
\end{lem}
\begin{proof}
Let $k$ be the value given by the \tp instance $I$ corresponding to a graph $G$
in the reduction set. Assume that $I$ has a solution. Then obviously cutting the
$m$ edges connecting the gadgets of $G$ gives a perfectly balanced solution to
$I$. Hence the optimal solution in this case has a cut size of at most $m$.
Accordingly the algorithm which approximates the cut size by a factor of
$\alpha$ will cut at most $\alpha m$ edges. We will show that in the other case
when $I$ does not have a solution, the algorithm will cut more than $\alpha m$
edges. Hence the algorithm can decide the \tp problem and thus the lemma
follows.

For the sake of deriving a contradiction assume that the algorithm cuts at most
$\alpha m$ edges in case the \tp instance $I$ corresponding to the given graph
$G$ does not permit a solution. Since $G$ is from a reduction set for \kbp, by
Definition~\ref{dfn:red-set} this means that from its $n$ vertices, in
total less than $p-\eps n$ are minority vertices in the colouring induced by the
computed solution of the algorithm. In particular this means that each gadget
$G_i$, where $i\in\{1,\ldots 3k\}$, of $G$ has a \emph{majority colour}, i.e.\ a
colour that at least half the vertices in $G_i$ share. This is because the size
of $G_i$ is $pa_i$ and we can safely assume that $a_i\geq 2$ (otherwise the \tp
instance is trivial). The majority colours of the gadgets induce a partition
$\mc{P}$ of the integers $a_i$ of $I$ into $k$ sets. That is, we introduce a set
in $\mc{P}$ for each colour and put an integer $a_i$ in a set if the majority
colour of $G_i$ corresponds the colour of the set.

Since we assume that $I$ does not admit a solution, if all sets in $\mc{P}$
contain exactly three integers there must be some sets for which the
contained integers do not sum up to exactly the threshold $s$. On the other hand
the bounds on the integers, assuring that $s/4 < a_i < s/2$ for each
$i\in\{1,\ldots, 3k\}$, mean that in case not all sets in $\mc{P}$ contain
exactly three elements, there must also exist sets for which the contained
numbers do not sum up to $s$. By the pigeonhole principle and the fact that the
sum over all $a_i$ equals $ks$, there must thus be some set $T$ among the $k$ in
$\mc{P}$ for which the sum of the integers is strictly less than $s$. Since the
involved numbers are integers we can conclude that the sum of the integers in
$T$ is in fact at most $s-1$. Therefore the number of vertices in the gadgets
corresponding to the integers in $T$ is at most $p(s-1)$. Let w.l.o.g.\ the
colour of $T$ be $1$. Apart from the vertices in these gadgets having
majority colour $1$, all vertices in $G$ that also have colour $1$ must be
minority vertices. Hence there must be less than $p(s-1)+p-\eps n$ many
vertices with colour $1$. Since $\sum_{i=1}^{3k} a_i=ks$ and thus $ps=n/k$,
these are less than $n/k-\eps n$.

At the same time the algorithm computes a solution inducing a colouring in which
each colour has at most $(1+\eps)n/k$ vertices, since $n$ is divisible by $k$.
This means we can give a lower bound of $n-(k-1)(1+\eps)n/k$ on the number of
vertices of a colour by assuming that all other colours have the maximum number
of vertices. Since this lower bound equals $(1+\eps)n/k-\eps n$, for any
$\eps\geq 0$ we get a contradiction on the upper bound derived above for colour
$1$. Thus the assumption that the algorithm cuts less than $\alpha m$ edges if
$I$ does not have a solution is wrong.
\end{proof}
}

\section{Consequences for Grids and Trees}

We will now consider solid grid graphs and trees to show the hardness of the
\kbp problem when restricted to these. For grids we establish our results
by considering a set of \emph{rectangular} grid graphs which are connected in a
row (Figure~\ref{fig:reduction}). By a rectangular grid graph we mean the
Cartesian product of two paths. That is, given two paths $P_1$ and $P_2$ with
vertex sets $V_1$ and $V_2$ respectively, their Cartesian product is a solid
grid graph with vertex set $V_1\times V_2$. Two vertices $(v_1,v_2)$ and
$(w_1,w_2)$ of the grid are adjacent if for some $i\in\{1,2\}$ there is an edge
between $v_i$ and $w_i$ in the path~$P_i$, while $v_j=w_j$ for the other index
$j\in\{1,2\}\setminus\{i\}$. Consider the natural planar embedding of a grid
graph where vertices are coordinates in $\mathbb{N}^2$ and edges have unit
length. The \emph{width} of a rectangular grid graph is the number of
vertices sharing the same $y$-coordinate in this embedding. Accordingly the
\emph{height} is the number sharing the same $x$-coordinate. We first prove that
such topologies can be used for reduction sets. We satisfy the conditions by
observing that a grid graph resembles a discretised polygon and hence shares
their isoperimetric properties. This fact was already used
in~\cite{PapadimitriouS96} and we harness these results in the following lemma.

\begin{figure}
\centering
\includegraphics[width=\textwidth]{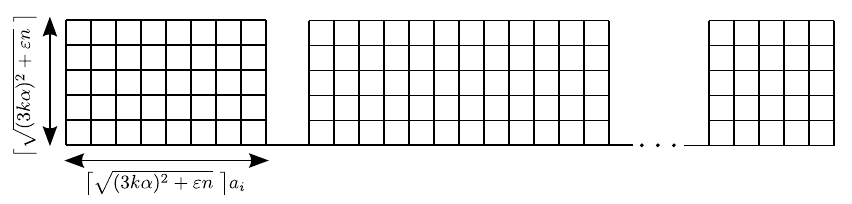} 
\caption{The solid grid constructed for the reduction from \tp. The gadgets
which are rectangular grids are connected through the bottom left and right
vertices.}
\label{fig:reduction}
\end{figure}

\begin{lem}\label{lem:grid-red-set}
Let $\eps\geq 0$ and $\alpha\geq 1$. For any \tp instance, let a solid grid
graph $G$ with $n$ vertices be given that consists of $3k$ rectangular grids
which are connected in a row using their lower left and lower right vertices by
$m=3k-1$ edges. Moreover let the height and width of a rectangular grid $G_i$,
where $i\in\{1,\ldots 3k\}$, be $\bigl\lceil\sqrt{(3k\alpha)^2+\eps n}\
\bigr\rceil$ and $\bigl\lceil\sqrt{(3k\alpha)^2+\eps n}\ \bigr\rceil a_i$,
respectively. If~$\eps$ and $\alpha$ are values for which these grids exist,
they form a reduction set for \kbp.
\end{lem}
\begin{proof}
Consider one of the described graphs $G$ for a \tp instance. Since both the
height and the width of each rectangular grid $G_i$ is greater than $\alpha m$,
using at most $\alpha m$ edges it is not possible to cut across a gadget~$G_i$,
neither in horizontal nor in vertical direction.
Due to~\cite[Lemma~2]{PapadimitriouS96} it follows that with this limited amount
of edges, the maximum number of vertices can be cut out from the gadgets by
using a square shaped cut in one corner of a single gadget. Such a cut will cut
out at most $(\alpha m/2)^2$ vertices. Hence if the vertices of the grid graph
$G$ are cut into $k$ sets using at most $\alpha m$ edges, then the induced
colouring contains at most $(\alpha m/2)^2$ minority vertices in total. Since
the size of each gadget is its height times its width, the parameter $p$ is
greater than $(\alpha m)^2+\eps n$. Hence the number of minority vertices is
less than~$p-\eps n$.
\end{proof}

The above topology is first used in the following theorem to show that no
satisfying fully polynomial time algorithm exists.

\begin{thm}\label{thm:Fpoly-grids}
Unless P=NP, there is no fully polynomial time algorithm for the \kbp problem on
solid grid graphs that for any~$\eps>0$ computes a solution in which each set
has size at most $(1+\eps)\lceil n/k \rceil$ and where $\alpha=n^c/\eps^d$, for
any constants $c$ and $d$ where $c<1/2$.
\end{thm}
\begin{proof}
In order to prove the claim we need to show that a reduction set as suggested by
Lemma~\ref{lem:grid-red-set} exists and can be constructed in polynomial time.
We first prove the existence by showing that the construction given by
Lemma~\ref{lem:grid-red-set} is feasible for any \tp instance. Since $\alpha$
and therefore the sizes of the gadgets depend on $n$, we need to determine the
number of vertices $n$ prior to the construction of the grid. The algorithm for
\kbp can compute a near-balanced partition for any $\eps>0$. Hence we can set
$\eps=(2ks)^{-1}$ so that $\alpha=n^c(2ks)^d$. For a solid grid graph as
suggested by Lemma~\ref{lem:grid-red-set} the parameter $p$ is determined by the
width and height of the gadgets. Hence $p=\bigl\lceil\sqrt{(3k\alpha)^2+\eps n}\
\bigr\rceil^2$ and the number of vertices is 
\begin{equation}\label{eq:n1}
n=\sum_{i=1}^{3k} pa_i=\left\lceil\sqrt{(3kn^c(2ks)^d)^2+n/(2ks)}\right\rceil^2
\cdot ks.
\end{equation}

Determining the non-zero solution for $n$ in this equation will give us the
number of vertices. For this we analyse the right-hand side of the equation as a
function of $n$ while ignoring the ceiling function. That is, we are interested
in the function
$$
f(n)=\left(\sqrt{(3kn^c(2ks)^d)^2+n/(2ks)}\right)^2\cdot ks=9k^3s(2ks)^{2d}
n^{2c}+n/2.
$$
It is easy to see that the points for which $f(n)=n$ are $n_0=0$ and
$n_1=(18k^3s(2ks)^{2d})^{{1}/(1-2c)}$. Setting $n$ to $n_1$ in the right-hand
side of Equation~\eqref{eq:n1} gives an upper bound $u\in\mathbb{N}$ on
$f(n_1)$,
which is integer valued due to the ceiling function and the
fact that $k,s\in\mathbb{N}$. Let $n_2$ be the value for which $f(n_2)=u$, which
is well-defined and greater than $n_1$ since $f(n)$ is strictly increasing. Note
that since $c<1/2$ the second derivative of the function $f(n)$ is negative.
That is, $f(n)$ is concave which, by definition of $n_0$ and $n_1$, means that
$f(n)\geq n$ for any $n\in[n_0,n_1]$ and also $f(n)\leq n$ for any
$n\notin[n_0,n_1]$. Consequently $f(n_2)\leq n_2$ from which we can conclude
that
$$
n_1 = f(n_1)\leq u=f(n_2)\leq n_2.
$$
Hence $u\in[n_1,n_2]$. Since $f(n)$ is strictly increasing, the right-hand
side of Equation~\eqref{eq:n1} equals $u$ for any $n\in[n_1,n_2]$. This means
that the non-zero solution to Equation~\eqref{eq:n1} is the value $u$.

By first calculating $u$, the construction of Lemma~\ref{lem:grid-red-set} can
be carried out given a \tp instance $I$.
Since $c$ and $d$ are constants and \tp is strongly NP-hard, $n_1$ is
polynomially bounded in the size of $I$. Clearly~$u$ is therefore also
polynomially bounded. Hence a grid graph as suggested by
Lemma~\ref{lem:grid-red-set} can be constructed in polynomial time. For
$\eps=(2ks)^{-1}$ a fully polynomial time algorithm has a running time that is
polynomial in the size of the instance $I$ when executed on the corresponding
grid. However, unless P=NP, this algorithm cannot exist since it decides the \tp
problem due to Lemma~\ref{lem:reduction}.
\end{proof}

Next we consider computing perfectly balanced partitions. Note that we may set
$\eps=0$ in Lemma~\ref{lem:reduction} for this purpose. The following theorem
shows the NP-hardness of \kbp on solid grids. It therefore considers running
times that are polynomial solely in the number of vertices and do not depend on
some input parameter $\eps$.

\begin{thm}\label{thm:NP-grids}
There is no polynomial time algorithm for the \kbp problem on solid grid graphs
that approximates the cut size within $\alpha=n^c$ for any constant $c<1/2$,
unless P=NP.
\end{thm}
\begin{proof}
We first need to show that a reduction set as proposed in
Lemma~\ref{lem:grid-red-set} exists in order to use it together with
Lemma~\ref{lem:reduction}. To prove the existence we determine the number of
vertices of a grid as suggested by Lemma~\ref{lem:grid-red-set}. If this number
is finite the construction of Lemma~\ref{lem:grid-red-set} is feasible. Since
the balance of the solution is not to be approximated we set $\eps=0$. The
parameter $p$ is determined by the height and width of the gadgets which in this
case are $\lceil 3k\alpha\rceil$ and $\lceil 3k\alpha\rceil a_i$, respectively,
for a gadget $G_i$. Hence the number of vertices of the resulting grid is
\begin{equation}\label{eq:n2}
n=\sum_{i=1}^{3k} pa_i\leq\lceil 3kn^c\rceil^2\cdot ks.
\end{equation}

Similar to the proof of Theorem~\ref{thm:Fpoly-grids}, the non-zero solution of
$n$ in this equation can be determined by considering the function
$f(n)=9k^3sn^{2c}$, i.e.\ the right-hand side of Equation~\eqref{eq:n2} ignoring
the ceiling function. The non-zero point $n_1$ at which $f(n)=n$ is
$n_1=(9k^3s)^{1/(1-2c)}$. The function $f(n)$ is again strictly increasing and
concave if $c<1/2$. Hence the same arguments as in the proof of
Theorem~\ref{thm:Fpoly-grids} can be used to show that the number of vertices is
the integer value $u$ obtained by setting $n$ to $n_1$ in the right-hand side of
Equation~\eqref{eq:n2}. Additionally it is polynomial in the size of the \tp
instance since $c$ is a constant and \tp is strongly NP-hard. Hence the grids
can be constructed in polynomial time given the value $u$ and the integers of a
\tp instance. By Lemma~\ref{lem:reduction}, a polynomial time algorithm which
computes a perfectly balanced partition on any grid given by the reduction set,
and which approximates the cut size within $\alpha$, can decide the \tp problem
in polynomial time. This gives a contradiction unless P=NP, which concludes the
proof.
\end{proof}

Lemma~\ref{lem:grid-red-set} shows that for solid grid graphs the hardness
derives from their isoperimetric properties. Trees do not experience such
qualities. However they may have high vertex degrees, which grids cannot. The
following theorem shows that this property also leads to a similar hardness as
for solid grid graphs. %

\begin{figure}
\centering
\includegraphics[width=0.8\textwidth]{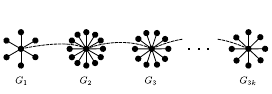} 
\caption{The tree constructed for the reduction from \tp. The gadgets
which are stars are connected through their centre vertices.}
\label{fig:tree-reduction}
\end{figure}

\begin{thm}\label{thm:trees}
Unless P=NP, there is no fully polynomial time algorithm for the \kbp problem on
trees that for any $\eps>0$ computes a solution in which each set has size at
most $(1+\eps)\lceil n/k \rceil$ and where $\alpha=n^c/\eps^d$, for any
constants $c$ and $d$ where $c<1$.
\end{thm}
\begin{proof}
We need to identify a reduction set for \kbp containing trees. We use very
similar gadgets to those used in~\cite[Theorem~2]{AFFoschini12}, despite the
fact that the proof idea in the latter paper is different from the one employed
in Lemma~\ref{lem:reduction}. Each gadget of a tree in our reduction set is a
star (Figure~\ref{fig:tree-reduction}) and these are connected in a path through
their centre vertices using $m=3k-1$ edges. The number of vertices in each star
$G_i$ is $pa_i$ where $p=\lceil 3k\alpha+\eps n\rceil$. Using at most $\alpha
m$, i.e.\ less than $3k\alpha$, edges to cut off vertices from a single star,
less than $3k\alpha$ leaves will be cut off. At the same time more than half the
vertices of the star are still connected to the centre vertex. This is because
each star contains at least $6k\alpha$ vertices since we can safely assume that
$a_i\geq 2$ for each $i\in\{1,\ldots,3k\}$ (otherwise the \tp instance is
trivial due to $s/4<a_i<s/2$). Therefore partitioning the vertices of all
gadgets into $k$ sets using at most $\alpha m$ edges will in total produce less
than $3k\alpha$ minority vertices in the induced colouring. This establishes the
desired upper bound on the number of minority vertices for the reduction set
since $3k\alpha\leq p-\eps n$.

In order to prove the claim we need to show that for the given parameters a
reduction set as suggested above exists and can be constructed in polynomial
time. We first prove the existence by determining the number of vertices in a
tree of the reduction set. Since the algorithm for \kbp can compute a
near-balanced partition for any $\eps>0$ we set $\eps=(2ks)^{-1}$. Thus the
number of vertices in a tree of the reduction set is 
\begin{equation}\label{eq:n3}
n=\sum_{i=1}^{3k}p a_i =\lceil 3kn^c(2ks)^d+n/(2ks)\rceil\cdot ks.
\end{equation}

As in the proof of Theorem~\ref{thm:Fpoly-grids} we can determine the non-zero
solution of~$n$ to this equality by considering the right-hand side and ignoring
the ceiling function. This gives the function $f(n)=3k^2s(2ks)^{d}n^c+n/2$. The
non-zero point $n_1$ at which $f(n)=n$ in this case is
$n_1=(6k^2s(2ks)^d)^{1/(1-c)}$. The function $f(n)$ is strictly increasing and
concave since $c<1$. Hence, similar to the proof of
Theorem~\ref{thm:Fpoly-grids}, we can determine the number of vertices by
setting~$n$ to $n_1$ in the right-hand side of Equation~\eqref{eq:n3}. Given a
\tp instance the construction of a tree of the reduction set can thus be carried
out by first calculating the number of vertices in the tree. Additionally,
if~$c$ and~$d$ are constants such a construction can be done in polynomial time
since \tp is strongly NP-hard.

For $\eps=(2ks)^{-1}$ a fully polynomial time algorithm has a running time that
is polynomial in the \tp instance when executed on the corresponding tree in the
reduction set. However, unless P=NP, this algorithm cannot exist since it
decides the \tp problem due to Lemma~\ref{lem:reduction}.
\end{proof}
\section{Conclusions}

Are there algorithms for the \kbp problem that are both \emph{fast} and compute
\emph{near-balanced solutions}, even when allowing the cut size to increase
when $\eps$ decreases? We gave a negative answer to this question. This means
that the tradeoff between fast running time and good solution quality, as
provided by the algorithms mentioned in the introduction, is unavoidable. In
particular the running time to compute near-balanced
solutions~\cite{AFFoschini12} has to increase exponentially with $\eps$. Also
our results draw a frontier of the hardness of the problem by showing how the
cut size needs to increase with decreasing $\eps$ for fully polynomial time
algorithms. These are the first bicriteria inapproximability results for the
\kbp problem.

To show the tightness of the achieved results, for grid graphs we harness
results by Diks {\it et al.}~\cite{DiksDSV93}. They provide a polynomial time
algorithm to cut out any number of vertices using at most $\mc{O}(\sqrt{\Delta
n})$ edges from a planar graph with maximum degree $\Delta$. Since $\Delta=4$
in a grid graph, it is possible to repeatedly cut out $\lceil n/k\rceil$
vertices from the grid using $\mc{O}(k\sqrt{n})$ edges in
total. A perfectly balanced partition needs at least $k-1$ edges in a connected
graph. Hence we obtain an $\alpha\in\mc{O}(\sqrt n)$ approximation algorithm for
grids.
This shows that both the hardness results we gave for these graphs are
asymptotically tight, since the algorithm runs in (fully) polynomial time. For
trees a trivial approximation algorithm can cut all edges in the graph and
thereby obtain $\alpha=n$. This shows that also the achieved result for trees is
asymptotically tight.

We were able to show that both trees and grids experience similar hardness. This
is remarkable since these graphs have entirely different combinatorial
properties. On the other hand, it emphasises the ability of the given reduction
framework to capture a fundamental trait of the \kbp problem. It remains to be
seen what other structural properties can be harnessed for our framework, in
order to prove the hardness for entirely different graph classes. Another
interesting approach would be to take the opposite view and identify
properties of graphs that in fact lead to good approximation algorithms.

It is interesting to note though that the isoperimetric properties that we used
to establish the hardness for grid graphs, can be further exploited to show the
hardness for related graph classes. For instance a graph class that is
independent of solid grid graphs (in that neither class contains the other) are
non-solid grids of simple shape. However even for the case when both the holes
and the grids are restricted to rectangular shapes, our reduction can be used to
give similar hardness results as for solid grid graphs. This can easily be seen
by adding edges to the gadgets used in our reduction set
(Figure~\ref{fig:reduction}) to connect the top most left and right vertices,
respectively. The constants of the respective parameters can readily be adapted
to compensate for the additional edges.

Another graph class for which some additional interesting observations can be
made are $\delta$-dimensional grids. Clearly the presented results are also
valid for these since solid grid graphs are a special case. However one
can show that the hardness in fact grows with the dimension $\delta$. For
example in a $3$-dimensional cuboid shaped grid, it is easy to see that the
maximum number of vertices that can be cut out using $\alpha m$ edges is in the
order of~$(\alpha m)^{3/2}$ if the grid is large enough. Such cuboid grids can
be used as gadgets for a reduction set. This leads to hardness results where the
constant~$c$ can take values up to $2/3$, i.e.\ the inverse of the former
exponent. By exploiting the isoperimetric properties of grids in higher
dimensions accordingly, one can show that the corresponding constant~$c$ can
take values up to~$1-1/\delta$.

Note that the respective ratios $\alpha$ of the bicriteria inapproximability
results can in each case be amplified arbitrarily due to the unrestricted
constant~$d$. Also, we are able to provide reduction sets for graphs that
resemble those resulting from 2D FEMs (solid grid graphs, or grid graphs with
holes of simple shape) and even 3D FEMs ($3$-dimensional grids) which are
widely used in practice. For these reasons our results imply that completely
different methods (possibly randomness or fixed parameter tractability) must be
employed in order to find fast practical algorithms with rigorously bounded
approximation guarantees.

\bibliographystyle{abbrv}
\bibliography{partitioning}

\begin{thebibliography}{10}

\bibitem{RackeA06}
K.~Andreev and H.~{R\"acke}.
\newblock Balanced graph partitioning.
\newblock {\em Theory of Computing Systems}, 39(6):929--939, 2006.

\bibitem{ArbenzLMMS07}
P.~Arbenz, G.~van Lenthe, U.~Mennel, R.~M{\"u}ller, and M.~Sala.
\newblock Multi-level $\mu$-finite element analysis for human bone structures.
\newblock In {\em Proceedings of the 8th Workshop on State-of-the-art in
  Scientific and Parallel Computing (PARA)}, pages 240--250, 2007.

\bibitem{AroraRV04}
S.~Arora, S.~Rao, and U.~Vazirani.
\newblock Expander flows, geometric embeddings and graph partitioning.
\newblock In {\em Proceedings of the 26th annual ACM symposium on Theory of
  computing (STOC)}, pages 222--231, 2004.

\bibitem{BhattL84}
S.~Bhatt and F.~T. Leighton.
\newblock A framework for solving {VLSI} graph layout problems.
\newblock {\em Journal of Computer and System Sciences}, 28(2):300--343, 1984.

\bibitem{ChevalierP08}
C.~Chevalier and F.~Pellegrini.
\newblock {PT-Scotch}: A tool for efficient parallel graph ordering.
\newblock {\em Parallel Computing}, 34(6–8):318--331, 2008.

\bibitem{DellingGPW11}
D.~Delling, A.~Goldberg, T.~Pajor, and R.~Werneck.
\newblock Customizable route planning.
\newblock {\em Experimental Algorithms}, pages 376--387, 2011.

\bibitem{DiestelJGT99}
R.~Diestel, T.~R. Jensen, K.~Y. Gorbunov, and C.~Thomassen.
\newblock Highly connected sets and the excluded grid theorem.
\newblock {\em Journal of Combinatorial Theory, Series B}, 75(1):61--73, 1999.

\bibitem{DiksDSV93}
K.~Diks, H.~N. Djidjev, O.~Sykora, and I.~Vrto.
\newblock Edge separators of planar and outerplanar graphs with applications.
\newblock {\em Journal of Algorithms}, 14(2):258 -- 279, 1993.

\bibitem{ElmanSW05}
H.~Elman, D.~Silvester, and A.~Wathen.
\newblock {\em Finite Elements and Fast Iterative Solvers: with Applications in
  Incompressible Fluid Dynamics}.
\newblock Oxford University Press, USA, 2005.

\bibitem{EvenNRS99}
G.~Even, J.~Naor, S.~Rao, and B.~Schieber.
\newblock Fast approximate graph partitioning algorithms.
\newblock {\em SIAM Journal on Computing}, 28(6):2187--2214, 1999.

\bibitem{AF12thesis}
A.~E. Feldmann.
\newblock {\em Balanced Partitioning of Grids and Related Graphs: A Theoretical
  Study of Data Distribution in Parallel Finite Element Model Simulations}.
\newblock PhD thesis, ETH Zurich, April 2012.
\newblock Diss.-Nr. ETH: 20371.

\bibitem{AF12}
A.~E. Feldmann.
\newblock Fast balanced partitioning is hard, even on grids and trees.
\newblock In {\em Proceedings of the 37th International Symposium on
  Mathematical Foundations of Computer Science (MFCS)}, pages 372--382, 2012.

\bibitem{AFDasWidmayer11}
A.~E. Feldmann, S.~Das, and P.~Widmayer.
\newblock Restricted cuts for bisections in solid grids: A proof via polygons.
\newblock In {\em Proceedings of the 37th International Workshop on
  Graph-Theoretic Concepts in Computer Science (WG)}, pages 143--154, 2011.

\bibitem{AFFoschini12}
A.~E. Feldmann and L.~Foschini.
\newblock Balanced partitions of trees and applications.
\newblock In {\em 29th International Symposium on Theoretical Aspects of
  Computer Science (STACS)}, pages 100--111, 2012.

\bibitem{AFWidmayer11}
A.~E. Feldmann and P.~Widmayer.
\newblock An {$O(n^4)$} time algorithm to compute the bisection width of solid
  grid graphs.
\newblock In {\em Proceedings of the 19th Annual European Symposium on
  Algorithms (ESA)}, pages 143--154, 2011.

\bibitem{GareyJ79}
M.~R. Garey and D.~S. Johnson.
\newblock {\em Computers and Intractability: A Guide to the Theory of
  NP-Completeness}.
\newblock W.H. Freeman and Co., 1979.

\bibitem{GareyJS76}
M.~R. Garey, D.~S. Johnson, and L.~Stockmeyer.
\newblock Some simplified {NP}-complete graph problems.
\newblock {\em Theoretical Computer Science}, 1(3):237--267, 1976.

\bibitem{KarypisK95}
G.~Karypis and V.~Kumar.
\newblock {METIS}-unstructured graph partitioning and sparse matrix ordering
  system, version 2.0.
\newblock Technical report, University of Minnesota, 1995.

\bibitem{Khot06}
S.~Khot.
\newblock Ruling out {PTAS} for graph min-bisection, dense $k$-subgraph, and
  bipartite clique.
\newblock {\em SIAM Journal on Computing}, 36(4):1025--1071, 2006.

\bibitem{KhotV05}
S.~A. Khot and N.~K. Vishnoi.
\newblock The {U}nique {G}ames {C}onjecture, integrality gap for cut problems
  and embeddability of negative type metrics into $\ell_1$.
\newblock In {\em Proceedings of the 46th Annual IEEE Symposium on Foundations
  of Computer Science (FOCS)}, pages 53--62, 2005.

\bibitem{KleinPR93}
P.~Klein, S.~Plotkin, and S.~Rao.
\newblock Excluded minors, network decomposition, and multicommodity flow.
\newblock In {\em Proceedings of the 25th Annual ACM Symposium on Theory of
  Computing (STOC)}, pages 682--690, 1993.

\bibitem{KrauthgamerNS09}
R.~Krauthgamer, J.~Naor, and R.~Schwartz.
\newblock Partitioning graphs into balanced components.
\newblock In {\em Proceedings of the 20th Annual ACM-SIAM Symposium on Discrete
  Algorithms (SODA)}, pages 942--949, 2009.

\bibitem{KwatraSETB03}
V.~Kwatra, A.~Sch{\"o}dl, I.~Essa, G.~Turk, and A.~Bobick.
\newblock Graphcut textures: Image and video synthesis using graph cuts.
\newblock {\em ACM Transactions on Graphics}, 22(3):277--286, 2003.

\bibitem{LeightonR99}
T.~Leighton and S.~Rao.
\newblock Multicommodity max-flow min-cut theorems and their use in designing
  approximation algorithms.
\newblock {\em Journal of the ACM}, 46(6):787--832, 1999.

\bibitem{LiptonT80}
R.~Lipton and R.~Tarjan.
\newblock Applications of a planar separator theorem.
\newblock {\em SIAM Journal on Computing}, 9:615--627, 1980.

\bibitem{MacGregor1978}
R.~M. {MacGregor}.
\newblock {\em On Partitioning a Graph: a Theoretical and Empirical Study.}
\newblock PhD thesis, University of California, Berkeley, 1978.

\bibitem{PapadimitriouS96}
C.~Papadimitriou and M.~Sideri.
\newblock The bisection width of grid graphs.
\newblock {\em Theory of Computing Systems}, 29:97--110, 1996.

\bibitem{ParkP93}
J.~K. Park and C.~A. Phillips.
\newblock Finding minimum-quotient cuts in planar graphs.
\newblock In {\em Proceedings of the 25th Annual ACM Symposium on Theory of
  Computing (STOC)}, pages 766--775, 1993.

\bibitem{Racke08}
H.~R{\"a}cke.
\newblock Optimal hierarchical decompositions for congestion minimization in
  networks.
\newblock In {\em Proceedings of the 40th Annual ACM Symposium on Theory of
  Computing (STOC)}, pages 255--264, 2008.

\bibitem{SimonT97}
H.~D. Simon and S.~H. Teng.
\newblock How good is recursive bisection?
\newblock {\em SIAM Journal on Scientific Computing}, 18(5):1436--1445, 1997.

\bibitem{LeahyW93}
Z.~Wu and R.~Leahy.
\newblock An optimal graph theoretic approach to data clustering: Theory and
  its application to image segmentation.
\newblock {\em IEEE Transactions on Pattern Analysis and Machine Intelligence},
  15(11):1101--1113, 1993.

\end{thebibliography}

\end{document}